\documentclass[sigconf,final]{acmart}

\copyrightyear{2023}
\acmYear{2023}
\setcopyright{rightsretained}
\acmConference[ISSAC 2023]{International Symposium on Symbolic and Algebraic Computation 2023}{July 24--27, 2023}{Troms\o, Norway}
\acmBooktitle{International Symposium on Symbolic and Algebraic Computation 2023 (ISSAC 2023), July 24--27, 2023, Troms\o, Norway}\acmDOI{10.1145/3597066.3597091}
\acmISBN{979-8-4007-0039-2/23/07}

\usepackage{amsmath}

\usepackage{mathtools}
\mathtoolsset{showonlyrefs,showmanualtags}

\usepackage{tikz}
\usepackage{enumitem}

\let\set\mathbb
\def\<#1>{\langle#1\rangle}

\def\i{\mathrm{i}}

\def\O{\operatorname{O}}

\DeclareMathOperator{\lclm}{lclm}
\DeclareMathOperator{\ord}{ord}
\DeclareMathOperator{\Span}{span}

\newtheorem{thm}{Theorem}
\newtheorem{prop}[thm]{Proposition}

\newtheorem{lem}[thm]{Lemma}
\newtheorem{ex}[thm]{Example}
\newtheorem{defi}[thm]{Definition}
\newtheorem{algo}[thm]{Algorithm}

\newtheorem{question}[thm]{Question}

\def\eatspace#1{#1}
\def\step#1#2{%
  \par\kern1pt\dimen144=#2em\advance\dimen144by1.67em
  \hangindent=\dimen144\hangafter=1
  \noindent\rlap{\small#1}\kern\dimen144\relax\eatspace}

\newcommand{\NN}{\set N}
\newcommand{\ZZ}{\set Z}
\newcommand{\QQ}{\set Q}

\begin{document}
\fancyhead{}
\title{Transcendence Certificates for D-finite Functions}

\thanks{M.~Kauers was supported by the Austrian FWF grants I6130-N and P31571-N32. C.~Koutschan  was supported by the Austrian FWF grant I6130-N. T.~Verron was supported by the Austrian FWF grants P31571-N32 and P34872-N.}

 \author[M. Kauers]{Manuel Kauers}
\affiliation{%
  \institution{Institute for Algebra\\Johannes Kepler University}
  \city{4040 Linz}
  \country{Austria}
}
\email{manuel.kauers@jku.at}

\author[C. Koutschan]{Christoph Koutschan}
\affiliation{%
  \institution{RICAM\\Austrian Academy of Sciences}
  \city{4040 Linz}
  \country{Austria}
}
\email{christoph.koutschan@oeaw.ac.at}

\author[T. Verron]{Thibaut Verron}
\affiliation{%
  \institution{Institute for Algebra\\Johannes Kepler University}
  \city{4040 Linz}
  \country{Austria}
}
\email{thibaut.verron@jku.at}

\begin{abstract}
  Although in theory we can decide whether a given D-finite function
  is transcendental, transcendence proofs remain a challenge in practice.
  Typically, transcendence is certified by
  checking certain incomplete sufficient conditions. In this paper
  we propose an additional such condition which catches some cases
  on which other tests fail. 
\end{abstract}


\maketitle

%
%

 

\section{Introduction}

 An algebraic function is a quantity $y$ for which there are
 polynomials $u_0,\dots,u_d$, not all zero, such that
 \[
   u_0(x) + u_1(x)y + \cdots + u_d(x)y^d = 0.
 \]
 A D-finite function is a quantity $y$ for which there are
 polynomials $p_0,\dots,p_r$, not all zero, such that
 \[
   p_0(x)y + p_1(x)y' + \cdots + p_r(x)y^{(r)} = 0.
 \]
 As recognized by Abel, every algebraic function is also D-finite, and it is not
 hard to construct a differential equation from a known polynomial equation.
 The other direction is much more difficult, as a given differential equation
 may or may not have any algebraic solutions.
 The problems of finding out whether a given differential equation has \emph{some} (nonzero)
 algebraic solutions, and finding out whether a \emph{given} power series solution
 of a given differential equation is algebraic can be reduced to the problem of
 finding out whether a given differential equation has \emph{only} algebraic solutions,
 using operator factorization~\cite{put03} or minimization techniques~\cite{bostan22},
 respectively.
 
 The problem to decide whether a given differential equation admits \emph{only}
 algebraic solutions has received a lot of attention since the 19th century,
 when Schwarz, Klein, Fuchs and others studied the problem for equations
 with $r=2$~\cite{gray86}, but even this special case was not fully understood until
 Baldassari and Dwork~\cite{baldassari79} gave a complete decision procedure in 1979.
 Only a year later, Singer~\cite{singer79} offered an algorithm that applies to equations
 of arbitrary order~$r$. His algorithm is, however, only of theoretical interest, as it
 relies on solving a nonlinear system of algebraic equations whose number of variables
 is determined by a group-theoretic bound involving the term $(49r)^{r^2}$. This is
 far from feasible, even for $r=2$.
 However, in practice, for small orders, the bound can be refined, leading to more practical algorithms.
 This has been done for order 2~\cite{kovacic86,SingerUlmer93}, order 3~\cite{SingerUlmer93,Ulmer05} and orders 4 and 5~\cite{Cormier01}.
 The problem remains difficult beyond those known cases.

 If a differential equation has only algebraic solutions, their minimal polynomials
 are not difficult to find. One way is to compute a truncated power series solution
 of the differential equation and then use linear algebra or Hermite-Pad\'e approximation~\cite{beckermann94}
 to find a candidate annihilating polynomial. From the first $N$ terms of a series
 solution, we can reliably detect annihilating polynomials of degrees $d_x,d_y$ with
 $(d_x+1)(d_y+1)<N$. The correctness of such a candidate can
 be checked by computing the differential equation satisfied by the solution of the
 candidate equation and comparing it with the input equation. If they do not match,
 or if no candidate equation is found, repeat the procedure with a higher truncation
 order $N$ and higher degrees $d_x,d_y$. Eventually, the correct minimal polynomial will be found.

 In Sect.~\ref{sec:expand-search-space} we give an alternative method which can decide
 for a given $d_y$ whether all solutions are algebraic with a minimal polynomial of degree
 at most~$d_y$, regardless of the degree $d_x$ of the polynomial coefficients of the
 minimal polynomial. This method has the advantage that $d_x$ need not be guessed in
 advance, but it still requires a guess for~$d_y$.
 We are thus led to the question how we can detect with a reasonable amount of
 computation time that a differential equation has at least one transcendental solution.
 There are indeed several things that are worth trying.
 For example, if a differential equation has a logarithmic or an exponential singularity,
 it cannot only have algebraic solutions.
 This test was applied for example in order to prove transcendence of the generating function
 for Kreweras walks with interacting boundaries~\cite{bostan21}.
 Another popular test is to determine the asymptotic behaviour of the series coefficients
 of a solution of the differential equation.
 If it is not of the form $\phi^n n^\alpha$ with $\alpha\in\set Q\setminus\{-1,-2,-3,\dots\}$,
 this also proves the presence of a transcendental solution~\cite{flajolet09}.
 A third possibility is to use arbitrary precision arithmetic~\cite{mezzarobba10a,kauers19c} to compute
 eigenvalues of monodromy matrices for the differential equation.
 If there is an eigenvalue that is not a root of unity, there must be a transcendental solution.
 A fourth idea is to exploit that an algebraic power series $f\in\set Q[[x]]$ must be \emph{globally bounded,} i.e.,
 there must be nonzero integers $\alpha,\beta$ such that $\alpha f(\beta x)\in\set Z[[x]]$. If a given
 differential operator has a series solution that is not globally bounded, then it cannot only
 have algebraic solutions.  
 As a fifth approach, we can investigate the $p$-curvature of the differential equation~\cite{bostan14a,bostan15a}
 and resort to a conjecture of Grothendieck according to which the $p$-curvature is zero for
 almost all primes~$p$ if and only if the differential equation has only algebraic solutions.
 A nice account on this approach was recently given by Bostan, Caruso, and Roques~\cite{bostan23}.
 Another idea is to try to prove transcendence via the criterion of Harris and Sibuya~\cite{harris85},
 which says that for a D-finite function~$f$, the reciprocal $1/f$ is D-finite as well if and only
 if the logarithmic derivative $f'/f$ is algebraic.
 Finally, there are powerful criteria for certain special differential equations, e.g., the
 criterion of Beukers and Heckman for testing algebraicity of a hypergeometric differential equation~\cite{beukers89}.

 All these tests have limitations. The first four tests only provide a sufficient condition
 for the existence of transcendental solutions, but there are equations with transcendental
 solutions on which all three tests fail. In addition, for the fourth test, even if we find
 a solution that looks like it is not globally bounded, it can be difficult to prove that it
 really is not. A limitation of the $p$-curvature test is the
 quantifier ``almost all'': if we encounter a prime (or several primes) for which the
 $p$-curvature is nonzero, this is strong evidence in favor of a transcendental solution,
 but there remains a small chance that the prime(s) were just unlucky.
 The criterion of Harris and Sibuya reduces the problem of proving that $f'/f$ is transcendental
 to the problem of proving that $1/f$ is not D-finite, which is typically more difficult.
 In fact, this criterion is more valuable in the other direction: to prove that $1/f$ is
 not D-finite, it suffices to prove that $f'/f$ is not algebraic.
 The obvious limitation of the criterion of Beukers and Heckman is that it only applies
 to hypergeometric functions. 
 
 In view of this situation, additional sufficient conditions for transcendental solutions
 that can be tested with reasonable computational cost are of interest. Ideally, such
 tests should also provide some artifacts that can serve as witness for the existence of
 transcendental solutions. We propose the term \emph{transcendence certificate} for such
 artifacts. For example, a logarithmic or exponential singularity can be viewed as such
 a transcendence certificate. Observe that the algorithms such as Singer's mentioned
 earlier do not provide any transcendence certificates but will just report ``no algebraic
 solution'' as output.

 The purpose of this paper is to introduce a transcendence certificate based on
 the following classical fact about algebraic functions:
 \begin{prop} \cite{vanDerWaerden31,bliss33}
   \label{prop:alghaspole}
   Every non-constant algebraic function must have at least one pole.
 \end{prop}
 With our new test, we are able to prove the existence of transcendental solutions
 for some equations that have no logarithmic singularities, no series solutions with illegal
 coefficient asymptotics, and whose monodromy matrices have just roots of unity as eigenvalues.
 We also wish to point out that our approach is applicable to differential equations of
 any order.
 
\section{Preliminaries}\label{sec:prelim}

Throughout this paper, let $C$ be an algebraically closed field of
characteristic zero, and let $K = C(x)$ denote the field of rational functions
over~$C$. A Puiseux series at~$\xi\in C$ is a series of the form
$c_n(x-\xi)^{n/q} + c_{n+1}(x-\xi)^{(n+1)/q}+\cdots$ with $n\in\set Z$,
$q\in\set N$, and $c_n,c_{n+1},\ldots\in C$; we write $C((\ (x-\xi)^{1/q}\ ))$
for the field of all Puiseux series at~$\xi$ whose exponents have a common
denominator dividing $q\in\set N$. Similarly, a Puiseux series at~$\infty$ is
a series of the form $c_nx^{-n/q} + c_{n+1}x^{-(n+1)/q}+\cdots =
c_n(x^{-1})^{n/q} + c_{n+1}(x^{-1})^{(n+1)/q}+\cdots$; the field of all
Puiseux series at~$\infty$ is denoted by $C((x^{-1/q}))$.  In both cases, we
call $n/q$ the \emph{starting exponent} of the series, provided that
$c_n\neq0$.

An algebraic function field $E=K[y]/\<m>$ is a field extension of the rational
function field~$K$ of finite degree, where $m$ is an irreducible polynomial
in~$K[y]$. For every $\xi\in C\cup\{\infty\}$, the element~$y\in E$ can be
identified with any of the $\deg_y(m)$ many roots of the minimal
polynomial~$m$ in the field of Puiseux series at~$\xi$; we call them the
expansions of $y$ at~$\xi$.

A Puiseux series is said to be \emph{integral} if its starting exponent is
nonnegative, i.e., if the corresponding function does not have a pole at the
expansion point.  The element $y$ of $E$ is called integral at $\xi\in
C\cup\{\infty\}$ if all its Puiseux series expansions at $\xi$ are integral.
In order to extend the definition of integrality to other elements of~$E$,
note that for every expansion $f$ of~$y$ we have a field homomorphism
$h_f\colon E\to C((\ (x-\xi)^{1/q}\ ))$ (or $h_f\colon E\to C((x^{-1/q}))$ if
$\xi=\infty$) which maps $y$ to~$f$. Now $u\in E$ is called integral at $\xi$
if for all expansions~$f$ of~$y$ the series $h_f(u)$ is integral. The element
$u$ is called (globally) \emph{integral} if it is integral at every $\xi\in C$ (but not
necessarily at infinity). The set of all integral elements of~$E$ forms a free
$C[x]$-submodule of~$E$, and a basis of this module is called an
\emph{integral basis} of~$E$. We say that an element of $E$ is \emph{completely integral}
if it is integral at every $\xi\in C\cup\{\infty\}$. According to Proposition~\ref{prop:alghaspole},
the completely integral elements of $E$ are precisely the elements of~$C$.

Let $D$ denote the usual derivation with respect to~$x$, i.e., $D(f)=f'$, which turns
$K=C(x)$ or $E=C(x)[y]/\<m>$ into differential fields. An element $c$ of a
differential field~$F$ is called a constant if $D(c)=0$; these constants
always form a subfield of~$F$.  A linear differential operator is an
expression of the form $L=p_0+p_1D + \cdots + p_rD^r$ with $p_0,\dots,p_r\in
K$. If $p_r\neq0$, we call $\ord(L)=r=\deg_D(L)$ the \emph{order} of the
operator.  The operator~$L$ is called \emph{monic} if $p_r=1$.  The set of all
linear differential operators will be denoted by~$K[D]$; it forms a
non-commutative ring in which the multiplication is governed by the Leibniz
rule $Dx=xD+1$. An operator~$L$ is called \emph{irreducible} if it cannot be
written as $L=L_1\cdot L_2$ with $\ord(L_1)\geq1$ and $\ord(L_2)\geq1$.  Every
differential field~$F$ is a $K[D]$-left-module via the action
\[
  (p_0+p_1D + \cdots + p_rD^r)\cdot y\ := \ p_0y+p_1D(y)+\cdots+p_rD^r(y).
\]
An element~$y$ of a differential field~$F$ is called a \emph{solution} of an
operator $L\in K[D]$ if $L\cdot y=0$. The set of all solutions of $L$ in a
differential field $F$ is denoted by $V(L)$.  It is always a vector space over
the constant field of~$F$ and hence called the \emph{solution space} of~$L$.
If the constant field of $F$ is~$C$, then the dimension of $V(L)$ in $F$ is
bounded by the order of~$L$, but in general it is smaller. We say that $L$ has
\emph{only algebraic solutions} if there is a differential field $E=K[y]/\<m>$
such that the solution space $V(L)$ in $E$ has dimension~$\ord(L)$.  If $L$ is
an irreducible operator then either all its solutions are algebraic or none of
them (except for the zero solution)~\cite[Prop.~2.5]{singer79}.

If $L=p_0+\cdots+p_rD^r\in K[D]$ is an operator of order~$r$, we call $\xi\in C$
a \emph{singularity} of $L$ if it is a pole of one of the rational functions
$p_0/p_r,\dots,p_{r-1}/p_r$. The point $\infty$ is called a singularity if,
after the substitution $x\mapsto x^{-1}$, the origin~$0$ becomes a singularity.
If $\xi\in C\cup\{\infty\}$ is not a singularity of~$L$, then $L$ has $r$
linearly independent Puiseux series solutions at~$\xi$, and they are all integral.

The notion of integrality for differential operators is defined
in a similar way as discussed above for algebraic field extensions $E=K[y]/\<m>$.
Throughout this paper, we consider only operators which have a basis of Puiseux series solutions at every point $\xi \in C \cup \{\infty\}$.
For such an operator $L\in K[D]$, we have the module $K[D]/\<L>$ where $\<L>$ denotes
the left ideal $\{P\cdot L\mid P\in K[D]\}$.
Note that $K[D]/\<L>$ is not a
ring but only a (left) $K[D]$-module. In this module, the equivalence class
$[1]_L$ has the property $L\cdot [1]_L=[L]_L=[0]_L$, so $[1]_L$ can be
considered as a solution of~$L$ in $K[D]/\<L>$, very much like the element
$y\in E$ is a root of~$m$. Similar as for algebraic function fields, we can
associate $[1]_L\in K[D]/\<L>$ with any solution~$f$ of $L$ in a Puiseux
series field $C((\ (x-\xi)^{1/q}\ ))$ or $C((x^{-1/q}))$. The association of
$[1]_L$ with $f$ extends to $K[D]/\<L>$ by mapping an equivalence class
$[P]_L$ to the series $P\cdot f$. The notions of integrality can now be
defined like before:
\begin{itemize}
\item $[P]_L$ is called (locally) \emph{integral} at some point $\xi\in C\cup\{\infty\}$
  if for every Puiseux series solution $f$ of $L$ at $\xi$, the series $P\cdot f$ is integral.
\item $[P]_L$ is called (globally) \emph{integral} if it is locally integral at every point $\xi\in C$
  (but not necessarily at $\infty$).  
\item $[P]_L$ is called \emph{completely integral} if it is locally integral at every point $\xi\in C\cup\{\infty\}$.
\end{itemize}
Note that in the last two items it suffices to consider points~$\xi$ that
are singularities of $L$ or poles of some of the coefficients of~$P$.
For any fixed $L$ and $P$, these are only finitely many.
Also recall that we restrict our attention to operators $L$ which have a basis of Puiseux solutions, so that the quantifier \emph{``for all Puiseux series solutions''} in the definitions above is equivalent to \emph{``for all solutions''}.

The set of all integral elements in $K[D]/\<L>$ forms a free $C[x]$-left-module,
and a basis of this module is called an \emph{integral basis} of $K[D]/\<L>$.
An integral basis $\{w_1,\dots,w_r\}$ is called \emph{normal at infinity} if
there are integers $\tau_1,\dots,\tau_r\in\set Z$ such that
$\{x^{\tau_1}w_1,\dots,x^{\tau_r}w_r\}$ is a basis of the
$C(x)_\infty$-left-module of all elements of $K[D]/\<L>$ which are integral at
infinity. Here, $C(x)_\infty$ refers to the ring of all rational functions
$u/v$ with $\deg u\leq\deg v$. Integral bases which are normal at infinity
always exist, and they can be computed~\cite{kauers15b,chen17a}.

Finally, we recall some fundamental facts about operators.
The \emph{adjoint} $L^\ast$ of an operator $L\in K[D]$ is defined in such a way
that for any two operators $L,M\in K[D]$ we have $(L+M)^\ast=L^\ast+M^\ast$ and
$(LM)^\ast=M^\ast L^\ast$. We have $D^\ast=-D$ and $q^\ast=q$ for all $q\in K$.
Moreover, $\ord(L^\ast)=\ord(L)$ for every $L\in K[D]$. 
The \emph{least common left multiple} of two operators $L,M\in K[D]$,
denoted by $\lclm(L,M)$, is defined as the unique monic operator of lowest order
which has both $L$ and $M$ as right factor. Its key feature is that whenever
$f$ is a solution of $L$ and $g$ is a solution of~$M$, then $f+g$ is a solution
of $\lclm(L,M)$.
For the efficient computation of the least common left multiple, see~\cite{bostan12b}.
There is a similar construction for multiplication. 
The \emph{symmetric product} $L\otimes M$ of two operators $L,M\in K[D]$ is
defined as the unique monic operator of lowest order such that whenever $f$
is a solution of $L$ and $g$ is a solution of~$M$, then $fg$ is a solution
of $L\otimes M$ (regardless of the differential field to which $f$ and $g$ belong).
As a special case, the $s$th \emph{symmetric power} of an operator $L\in K[D]$
is defined as $L^{\otimes s}=L\otimes\cdots\otimes L$.
For the efficient computation of the symmetric powers, see~\cite{bronstein97a}.

By construction, we have $V(L)+V(M)\subseteq V(\lclm(L,M))$, and in general, the
inclusion is proper. However, if $\dim V(L)=\ord(L)$ and $\dim V(M)=\ord(M)$,
then we have $V(L)+V(M)=V(\lclm(L,M))$, i.e., the least common multiple cannot have
any extraneous solutions. Likewise, if $\dim V(L)=\ord(L)$ and $\dim V(M)=\ord(M)$,
the solution space of the symmetric product $L\otimes M$ is generated by all products
$fg$ with $f\in V(L)$ and $g\in V(M)$. These facts were shown by Singer~\cite{singer79}
in the context of complex functions, and again using more abstract machinery in the
book of van der Put and Singer~\cite{put03}.

\section{Pseudoconstants}

Let $L \in K[D]$ be a linear differential operator. As mentioned before, if $L$ has a logarithmic
or exponential singularity, it follows immediately that $L$ does not only have
algebraic solutions and we may view the singularity as a transcendence certificate.
We continue to exclude this case from consideration, i.e., we continue to assume
that $L$ has no logarithmic or exponential singularity at any point in $C\cup\{\infty\}$.
In other words, we assume that $L$ has a basis of Puiseux series solutions at
every point.

\begin{defi}\label{def:pseudoconstants}
  Let $L \in K[D]$, and let $[P]_{L} \in K[D]/\langle L \rangle$.
  \begin{enumerate}[leftmargin=*]
    \item $[P]_{L}$ is called a \emph{constant} if $D \cdot [P]_{L} = [0]_{L}$;
    \item $[P]_{L}$ is called a \emph{pseudoconstant} if $[P]_{L}$ is completely integral
      but not a constant.
  \end{enumerate}
  We will say for short that ``$L$ has a [pseudo]constant'' if $K[D]/\<L>$ contains a [pseudo]constant. 
\end{defi}

\begin{prop}
  \label{prop:constant}
  Let $L \in K[D]$, and let $[P]_{L} \in K[D]/\langle L \rangle$.
  Let $E$ be an extension of $K$ such that the solution space $V(L)$ of $L$ in $E$ has dimension~$\ord(L)$. 
  \begin{enumerate}[leftmargin=*]
  \item\label{prop:constant:1} $[P]_{L}$ is a constant if and only if $P \cdot f$ is a constant for every $f\in V(L)$.
  \item If $[P]_L$ is a nonzero constant and $\ord(P)<\ord(L)$, then $\ord(P)=\ord(L)-1$. 
  \item The set of all constants forms a $C$-vector space of dimension at most~$\ord(L)$.
  \end{enumerate}
\end{prop}
\begin{proof}
  \quad 
  \begin{enumerate}[leftmargin=*]
  \item 
    Clearly, if $[P]_{L}$ is a constant, then for all $f \in V(L)$, $D\cdot (P
    \cdot f) = (D \cdot [P]_{L}) \cdot f = 0$. Conversely, let $r$ be the order
    of $L$ and $P$ be the representative of order at most $r-1$ of $[P]_{L}$.
    Assume that $P \cdot f$ is a constant for all $f \in V(L)$, i.e., $D \cdot
    (P \cdot f) = 0$. This means that $V(L) \subset V(D\cdot P)$. Since $V(D
    \cdot P)$ has dimension at most~$r$ and $V(L)$ has dimension~$r$, it follows
    that $V(L) = V(D \cdot P)$. This implies that $L$ and $D\cdot P$ are equal
    up to an invertible factor in $K$, and therefore that $D \cdot [P]_{L} =
    [D\cdot P]_{L} = [0]_{L}$.
  \item If $\ord(P)<\ord(L)-1$, then $\ord(DP)<\ord(L)$, so the assumption
    $D\cdot[P]_L=[DP]_L=0$ forces $DP=0$, which in turn forces $P=0$ in
    contradiction to the assumption that $[P]_L$ is not zero.
  \item It is clear that the constants form a $C$-vector space. In order to
    prove the bound on the dimension, consider a $P\in K[D]$ with $\ord(P)<\ord(L)$
    such that $[P]_L$ is a constant. 
    Then $D\cdot[P]_L=[DP]_L=0$, so there is a $q\in K$ with $DP=qL$. It is
    clear that $q$ is uniquely determined and that the function which maps every
    constant $[P]_L$ to the corresponding $q$ is $C$-linear and injective. Now
    $DP=qL$ implies $(DP)^\ast=(qL)^\ast$, so $P^\ast D^\ast=L^\ast q^\ast$, so
    $-P^\ast D=L^\ast q$. Since $1$ is a solution of the left hand side, it must
    be a solution of the right hand side, so $0=(L^\ast q)\cdot1=L^\ast\cdot q$,
    so $q\in V(L^\ast)$. We have thus constructed an injective $C$-linear map
    from the space of all constants to the solution space of $L^\ast$ in $K$.
    Since the dimension of the latter is at most $\ord(L)$, the claim follows.
    \qedhere
  \end{enumerate}
\end{proof}

If $[P]_{L}$ is a constant, then it is completely integral, but unlike in the case of algebraic
functions, the converse is not true in general. This means that pseudoconstants may exist.

\begin{ex}
  Let $L = 3x(x^{2}-1)D^{2} + 2(3x^{2}-1)D$.
  All its solutions are integral at every place including infinity,
  therefore $[1]_{L}$ is completely integral. 
  However, $D\cdot [1]_{L} = [D]_{L} \neq [0]_{L}$, so it is not a constant.
  Alternatively, one can observe that $L$ has a non-constant solution, and therefore $[1]_{L}$ cannot be a constant.
  So $[1]_{L}$ is a pseudoconstant. 
\end{ex}

In view of Prop.~\ref{prop:alghaspole}, we can regard pseudoconstants as transcendence certificates.

\begin{thm}
  Let $L \in K[D]$ be such that there exists a pseudoconstant $[P]_{L} \in K[D]/\langle L \rangle$.
  Then $L$ admits at least one transcendental solution.
\end{thm}
\begin{proof}
  For a contradiction, assume that $L$ has only algebraic solutions.
  Let $E$ be an algebraic extension of $K$ such that the solution space $V(L)$ in $E$ has dimension $\ord(L)$.
  Since algebraic functions are closed under application of linear operators, $P\cdot f$ is algebraic for all $f \in V(L)$.
  Since $[P]_{L}$ is completely integral, $P \cdot f$ does not have a pole at any $\xi\in C\cup\{\infty\}$.
  By Prop.~\ref{prop:alghaspole}, this implies that $P \cdot f$ is constant.
  Therefore, by Prop.~\ref{prop:constant}, $[P]_{L}$ is a constant, which is a contradiction.
\end{proof}
 
\begin{ex}
  \label{ex:product_2F1}
  \allowdisplaybreaks
  Consider the operator
  \begin{equation}
    \label{eq:2}
    \textstyle
    L = \bigl(x^{2} - x\bigr) D^{2} + \bigl(\frac{31}{24} x - \frac{5}{6}\bigr) D + \frac{1}{48},
  \end{equation}
  annihilating the function
  $x^{1/6}(x-1)^{13/24}{}_{2}F_{1}\bigl(\frac{7}{8},\frac{5}{6}; \frac{7}{6}; x\bigr)$.
  The operator is irreducible, and therefore all its solutions have the same nature.
  By Schwarz' classification and closure properties, they must be transcendental, but let us ignore this argument
  for the sake of the example.
  
  The singularities of the operator are $0$, $1$ and $\infty$, and a basis of solutions at each singularity is given by
  \def\xinf{1/x}
  \def\expfrac#1#2{#1/#2}
  \begin{align}
    \label{eq:7}
    y_{0,1} &= \textstyle x^{\expfrac{1}{6}} \Big( 1 + \frac{1}{12}x + \O(x^{2})\Big) \\
      y_{0,2} & = \textstyle 1 + \frac{1}{40}x + \O(x^{2}) \\
    y_{1,1} &= \textstyle (x{-}1)^{\expfrac{13}{24}} \Big( 1 - \frac{34}{111}(x{-}1) + \O((x{-}1)^{2})\Big) \\
      y_{1,2} & = \textstyle 1 - \frac{1}{22}(x{-}1) + \O((x{-}1)^{2}) \\
    y_{\infty,1} &= \textstyle (\xinf)^{\expfrac{1}{6}} \Big( 1 + \frac{4}{75}\left( \xinf \right) + \O(\left(\xinf\right)^{2})\Big) \\
      y_{\infty,2} & = \textstyle (\xinf)^{\expfrac{7}{8}} \Big( 1 + \frac{7}{184}\left( \xinf \right) + \O(\left( \xinf\right)^{2})  \Big)
  \end{align}
  Therefore, $[1]_L$ is a pseudoconstant, and thus the operator $L$ has no nonzero algebraic solution.

  As noted in the introduction, we could also compute the monodromy matrices of $L$ around $0$, $1$ and $\infty$.
  If one of them was not a root of unity, this would give another proof of transcendence.
  However, numeric computations suggest that all eigenvalues are roots of unity in this example.
  More precisely, the monodromy group around $0$ is generated by two matrices $M_1$ and $M_2$ with
  \begin{equation}
    \label{eq:11}
    M_1^{3} =
    \begin{pmatrix}
      1 & 0 \\ 0 & 1
    \end{pmatrix}
    \pm 10^{-17}
    \begin{pmatrix}
      0 & 0 \\
      0 & 7.38 \pm 6.75\i
    \end{pmatrix}
  \end{equation}
  and
  \begin{equation}
    \label{eq:12}
    M_2^{24} =
    \begin{pmatrix}
      1 & 0 \\ 0 & 1
    \end{pmatrix}
    \pm 10^{-13}
    \begin{pmatrix}
      1.45 \pm 1.42\i & 3.44 \pm 3.42\i \\
      0.758 \pm 0.757\i & 1.96 \pm 1.96\i
    \end{pmatrix}
  \end{equation}
  At $1$, the monodromy group is generated by two 6th roots of unity, and at $\infty$, by two 24th roots of unity.
\end{ex}

\begin{ex}
  \label{ex:pseudoconstant-ord3}
  Consider the operator
  \begin{align*}
    L ={}& (x - 1)^3 x^3  (x + 1)^3 D^{3} \\
    & \textstyle + \frac{19}{5}  (x - 1)^{2} x^{2}  (x+1)^{2} \bigl(x^{2} + \frac{22069}{9576} x - \frac{195}{152}\bigr) D^{2}\\
    & \textstyle -\frac{99}{80}  (x - 1) x (x + 1)  \\
    & \hphantom{{}+{}}\textstyle\bigl(x^{4} - \frac{117001919}{37422} x^{3} - \frac{105923}{5346} x^{2} + \frac{16795789}{5346} x + \frac{205}{66}\bigr) D
      \\
    & \textstyle -\frac{9}{20} x^{6} + \frac{517319279}{68040} x^{5} + \frac{256382531}{27216} x^{4} \\
    & \hphantom{{}+{}}\textstyle - \frac{19723513}{4320} x^{3} - \frac{2560752251}{272160} x^{2} - \frac{828238469}{272160} x - \frac{3}{32}.
  \end{align*}
  This operator has the singularities $0,1,-1,\infty$, with respective initial exponents
  \begin{equation}
    \label{eq:4}
    \begin{array}{rccc}
      (0) & -\frac{1}{8} & -\frac{3}{4} & -1 \\[1ex]
      (1) & \frac{5}{7} & \frac{4}{9} & -2 \\[1ex]
      (-1) & \frac{5171}{630} & \frac{3}{8} & -\frac{2}{3} \\[1ex]
      (\infty) & \frac{4}{5} & \frac{3}{4} & -\frac{3}{4} 
    \end{array}
  \end{equation}
  The operator is irreducible, and therefore all its solutions have the same nature.
  $L$ has the pseudoconstant $[P]_L$, with
\begin{align}
    \label{eq:6}
    P = {} &
      (x +1)^{-6} x^{3}  (x - 1)^{2} D^{2} \\
    & \textstyle +  (x + 1)^{-7} x^{2} (x - 1) \alpha(x) D \\
    & + (x + 1)^{-8} x \beta(x),
  \end{align}
where $\alpha(x)$ and $\beta(x)$ are certain polynomials of degree 3 and~6 respectively,
with coefficients in $\QQ$. So all the solutions of $L$ are transcendental.
\end{ex}

For operators with at most $3$ singularities, the nature of the solutions and the existence of pseudoconstants are determined by the initial exponents of the solutions.
Indeed, the operator is then uniquely determined up to a scalar factor by its singularities and initial exponents. Changing the position of the singularities is equivalent to applying a rational change of variables by a M\"obius transform, which preserves the nature of the solutions and the pseudoconstants.

This property does not hold for operators with more singularities, as the next example shows.

\begin{ex}
  Consider the operator
  \begin{align}
    \label{eq:10}
    L = {} & (x - 2)^{3}  (x - 1)^{3}  x^{3} D^{3} \\
    &\textstyle {} + \frac{19}{5}  (x - 2)^{2}  (x - 1)^{2}  x^{2}   \bigl(x^{2} - \frac{16547}{9576} x + \frac{2420}{1197}\bigr) D^{2} \\
    & \textstyle {}+ \frac{99}{80}  (x - 2)  (x - 1)  x    \\
    &\textstyle \hphantom{{}+{}} \bigl(x^{4} + \frac{8816399}{112266} x^{3} - \frac{8566381}{37422} x^{2} + \frac{7980386}{56133} x - \frac{3200}{6237}\bigr) D \\
    & \textstyle {} -\frac{9}{20} x^{6} + \frac{5640547}{68040} x^{5} - \frac{20050393}{136080} x^{4} \\
    & \textstyle {}- \frac{2904319}{30240} x^{3} + \frac{5167531}{54432} x^{2} + \frac{1144387}{19440} x + \frac{320}{63}.
  \end{align}
  It has the singularities $0,1,2,\infty$, with respective initial exponents:
    \begin{equation}
    \label{eq:4b}
    \begin{array}{rccc}
      (0) & \frac{5}{7} & \frac{4}{9} & -2 \\[1ex]
      (1) & \frac{5171}{630} & \frac{3}{8} & -\frac{2}{3} \\[1ex]
      (2) & -\frac{1}{8} & -\frac{3}{4} & -1 \\[1ex]
      (\infty) & \frac{4}{5} & \frac{3}{4} & -\frac{3}{4} 
    \end{array}
  \end{equation}
  The initial exponents are the same as those in Example~\ref{ex:pseudoconstant-ord3}, but the position of the singularities differ.
  Unlike the operator in Example~\ref{ex:pseudoconstant-ord3}, the operator $L$ does not admit a pseudoconstant.
  Note that using the technique described in~\cite{SingerUlmer93}, it can be proven that the operator $L$ does nonetheless admit only transcendental solutions.
\end{ex}

\begin{ex}
  In order to illustrate that this proof technique works for operators of any order, we provide%
  \footnote{\url{https://github.com/mkauers/ore_algebra/blob/master/src/ore_algebra/examples/pseudoconstants.py}}
   an operator of order 6 as well as a transcendence certificate.
  The operator has singularities at $0, 1, \dots, 6$ as well as $\infty$, with the following exponents:
  \begin{equation}
    \begin{array}{rcccccc}
      (0, \dots, 6) & -\frac{2}{7} & 0 & \frac{3}{7}& \frac{8}{7}& \frac{13}{7}& \frac{18}{7} \\[1ex]
      (\infty) & -1 & 0 & 1 & 2 & 3 & 4
    \end{array}
  \end{equation}
\end{ex}

There are at least two ways to search for pseudoconstants for a given~$L$.
The first one uses integral bases. It is shown in Lemma~8 of \cite{chen17a} that
a basis of the $C$-vector space of all completely integral elements of $K[D]/\<L>$
is given by $\{\,x^jw_i : i=1,\dots,r; j=0,\dots,\tau_i\,\}$ whenever
$\{w_1,\dots,w_r\}$ is an integral basis that is normal at infinity and
$\tau_1,\dots,\tau_r\in\set Z$ are such that $\{x^{\tau_1}w_1,\dots,x^{\tau_r}w_r\}$
is a local integral basis at infinity. This motivates the following algorithm.

\begin{algo}
  \label{algo:pseudoconstants_integral_basis}
  Input: $L \in K[D]$
  
  Output: a pseudoconstant of $L$ if there is one, otherwise $\bot$.

  \step 10 Compute an integral basis $w_{1},\dots,w_{r}$ of $K[D]/\langle L\rangle$ which
    is normal at~$\infty$, and the corresponding $\tau_{1},\dots,\tau_{r} \in \ZZ$
  \step 20 If there are $i\in\{1,\dots,r\}$ and $j\in\{0,\dots,\tau_i\}$ with
    $[Dx^jw_i]_L\neq0$, return such an $x^jw_i$
  \step 30 Otherwise, return~$\bot$
\end{algo}

\begin{thm}
  Algorithm~\ref{algo:pseudoconstants_integral_basis} is correct. 
\end{thm}
\begin{proof}
  It is clear that the algorithm is correct if it does not return~$\bot$.
  It remains to show that $L$ has no pseudoconstant if the algorithm does return~$\bot$.
  In view of the remarks before the algorithm, every completely integral element of
  $K[D]/\<L>$, and thus in particular every pseudoconstant, is a $C$-linear combination
  of the~$x^jw_i$. But if all the $x^jw_i$ were constants, then, since the constants
  also form a $C$-vector space, so would be all their linear combinations. Therefore,
  if there are pseudoconstants at all, there must be one among the $x^jw_i$.
\end{proof}

\def\myceil#1{\lceil -#1 \rceil}%
An implementation of Algorithm~\ref{algo:pseudoconstants_integral_basis} is available in the latest version of the SageMath package \texttt{ore\_algebra}\footnote{\url{https://github.com/mkauers/ore_algebra}}.
Otherwise, in an environment where no functionality for computing integral bases is available, we can use linear
algebra to search for pseudoconstants by brute force. This has the advantage of being conceptually
more simple, but the disadvantage that we cannot easily recognize the absence of pseudoconstants.
Let $\xi_{1},\dots,\xi_{m}\in C$ be the singularities of~$L$, and assume that $\infty$ is not a singularity.
At each singularity~$\xi_{i}$, let $\frac{p_{i}}{q} \in \QQ$ be the smallest exponent appearing in
one of the solutions at~$\xi_{i}$.
Let $u = (x-\xi_{1})^{\max(0,\myceil{p_{1}/q})}\cdots (x-\xi_{m})^{\max(0,\myceil{p_{m}/q})}$, so
that $[u]_{L}$ is globally integral.

For each singularity $\xi_{i}$, choose a bound $N_{i} \in \NN$ on the degree of the denominator of a
local integral basis at~$\xi_{i}$, and let $N = N_1+\dots+N_{m}$.

We form the ansatz
\begin{equation}
  q = \frac{u}{(x-\xi_{1})^{N_{1}}\cdots (x-\xi_{m})^{N_{m}}} \sum_{j=0}^{r-1} \sum_{i=0}^{N} c_{i,j} x^{i}D^{j}.
  \label{eq:1}
\end{equation}
with unknowns $c_{i,j}$.
Evaluating it at all solutions at $\xi_{1},\dots,\xi_{m},\infty$ gives series whose coefficients are linear combinations of the unknowns $c_{i,j}$, and setting those coefficients with negative valuations to $0$ yields a system of linear equations to solve.
Each solution is an operator which is completely integral.

However, if no non-zero solution is found, or if all solutions are constants, this is not enough to conclude that the operator does not have a pseudoconstant. It could just mean that the guessed bounds on the denominator were too conservative.



If $L$ does not have a pseudoconstant, we could try to apply some transformation to $L$ that does not change
the nature of the solutions of $L$ but may affect the existence of pseudoconstants.
For example, applying a gauge transform to $L$ does not change the nature of its solutions.
However, gauge transforms do not affect the existence of pseudoconstants either.
Indeed, let $L \in K[D]$ be a linear operator, $M \in K[D]$ be another one and $L'$ be the gauge transform of $L$
such that $V(L') = \{M \cdot f : f \in V(L)\}$. Assume that $[P]_{L'}$ is a pseudoconstant in $K[D]/\<L'>$.
Then $PM \cdot f$ does not have a pole for any $f\in V(L)$, and there exists an $f\in V(L)$ such that $PM\cdot f$
is not a constant. By definition, this implies that $[PM]_{L}$ is a pseudoconstant in $K[D]/\<L>$.
In conclusion, gauge transforms are not strong enough to create pseudoconstants.
We will see next that we may have more success with other operations. 

\section{Symmetric powers}
\label{sec:expand-search-space}

Symmetric powers are useful for proving identities among D-finite functions and they find
applications in algorithms for factoring operators~\cite{put03}.
They can also be used to decide for a given operator $L$ and a given $d\in\set N$
whether all solutions of $L$ are algebraic functions of degree at most~$d$.
For, if $f$ is an algebraic solution of $L$ with a minimal polynomial $m\in K[y]$ of degree~$d$,
then $m$ has $d$ distinct solutions $f_1,\dots,f_d$ in an algebraic closure $\bar K$ of~$K$
and we can write $m=(y-f_1)\cdots(y-f_d)$.
The solutions $f_1,\dots,f_d$ of $m$ are conjugates of~$f$, and since $L$ has coefficients in~$K$,
we have $L\cdot\sigma(f)=\sigma(L\cdot f)=0$ for every automorphism $\sigma$ that fixes~$K$.
Therefore, $f_1,\dots,f_d$ are also solutions of~$L$.
For every~$i$, the $i$th coefficient of $m = (y-f_{1})\cdots (y-f_{d})$ is the $(d-i)$th elementary symmetric
polynomial of $f_1,\dots,f_d$ and therefore an element of $L^{\otimes(d-i)}$.
As the coefficients of $m$ belong to~$K=C(x)$, they must show up among the rational solutions
of $L^{\otimes(d-i)}$. This observation motivates the following algorithm.

\begin{algo}\label{alg:algsols}
  Input: $L\in C(x)[D]$ and $d\in\set N$.

  Output: if all solutions of $L$ are algebraic functions of degree at most $d$, the minimal polynomial of one such
  solution; otherwise~$\bot$.

  \step 10 for $i=1,\dots,d$, compute the symmetric power $L^{\otimes i}$.
  \step 20 for $i=1,\dots,d$, compute basis elements $q_{i,1},\dots,q_{i,N_i}$ of the solution space of $L^{\otimes i}$
  in $C(x)$.
  \step 30 form an ansatz $y^{d} + \sum_{i=1}^{d} \sum_{j=1}^{N_{i}} c_{i,j}q_{i,j}y^{d-i}$ with undetermined
  coefficients $c_{i,j}$
  \step 40 substitute a truncated series solution $f$ of $L$ into the ansatz, equate coefficients, and solve the resulting system for the undetermined coefficients $c_{i,j}$.
  \step 50 if the system has no solution, return $\bot$.
  \step 60 let $m$ be the polynomial corresponding to one of the solutions of the linear system.
  \step 70 if all roots of $m$ are solutions of~$L$, return $m$
  \step 80 otherwise, go back to step~4 and try again with a higher truncation order.
\end{algo}

Compared to the guess-and-prove approach mentioned in the introduction, the algorithm above has the advantage
that only one of the degrees of the minimal polynomials has to be guessed.

Algorithm~\ref{alg:algsols} indicates that symmetric powers know something about algebraicity of solutions.
The next result points in the same direction.
It says that the symmetric powers of an operator $L$ are larger if $L$ has a transcendental solution. 

\begin{thm}
  Let $L \in C(x)[D]$.
  \begin{enumerate}[leftmargin=*]
    \item If $L$ has only algebraic solutions, then $\ord(L^{\otimes s}) = \O(s)$ as $s \to \infty$.
    \item If $L$ has at least one transcendental solution and $D^{2}$ is a right factor of $L$, then $\ord(L^{\otimes s}) = \Omega(s^{2})$ for $s \to \infty$.
  \end{enumerate}
\end{thm}
\begin{proof}
  Let $r$ be the order of~$L$.
  \begin{enumerate}[leftmargin=*]
   \item Let $f_1,\dots,f_r$ be a basis of~$V(L)$, and let $m_1,\dots,m_r\in C(x)[y]$ be their respective
     minimal polynomials. Furthermore, let $I_{\mathrm{rat}} = \{\,p\in C(x)[y_1,\dots,y_r] : p(f_1,\dots,f_r)=0\,\}$ be
     the ideal of algebraic relations among $f_1,\dots,f_r$.
     Since $m_i(y_i)\in I_{\mathrm{rat}}$, we have $\dim(I_{\mathrm{rat}})=0$.
     Therefore, the ideal $I_{\mathrm{pol}}=I_{\mathrm{rat}}\cap C[x][y_1,\dots,y_r]$ has dimension~1.
     As eliminating a variable cannot increase the dimension, we find that the ideal
     $I_{\mathrm{const}}:=I_{\mathrm{pol}}\cap C[y_1,\dots,y_r]$ has dimension at most~1.
     This means that the dimension of the $C$-vector space
     generated in $C[y_1,\dots,y_r]/I$ by the power products $y_1^{e_1}\cdots y_r^{e_r}$
     with $e_1,\dots,e_r\in\set N$ such that $e_1+\cdots+e_r\leq s$ has dimension $\O(s^1)$, as $s\to\infty$.
     Therefore, the dimension of the $C$-vector space generated by $f_1^{e_1}\cdots f_r^{e_r}$ with $e_1,\dots,e_r\in\set N$
     such that $e_1+\cdots+e_r=s$ has dimension $\O(s^1)$, as $s\to\infty$.
     This space is the solution space of $L^{\otimes s}$, and the order of $L^{\otimes s}$
     matches the dimension of this space.
   \item Since $D^2$ is a right factor of~$L$, we have $1$ and $x$ among the solutions of~$L$. If there
     is also at least one transcendental solution~$f$, then the solution space of $L^{\otimes s}$ contains
     all elements $1^{e_1}x^{e_2}f^{e_3}$ with $e_1,e_2,e_3\in\set N$ such that $e_1+e_2+e_3=s$, and the
     transcendence of $f$ implies that they are all linearly independent over~$C$.
     As these are $\binom{s+2}s=\Omega(s^2)$ many, the claim follows again from
     $\dim_C V(L^{\otimes s})=\ord(L^{\otimes s})$. \qedhere
   \end{enumerate}
\end{proof}

This theorem provides yet another heuristic test for the existence of transcendental solutions:
simply compute $L^{\otimes s}$ for the first
few $s$ and see how their orders grow. As the theorem only makes a statement for asymptotically large~$s$, looking at specific
values of $s$ will not allow us to make any definite conclusion, but it can provide convincing evidence. 

\begin{ex}
  Consider the operators
  \begin{align}
    \label{eq:3}
    L_{1} & =  \bigl(256x^5-3125\bigr)D^{4} + 3200x^{4} D^{3} \\
    &\hphantom{{}={}}+ 9840x^{3}D^{2} + 6120 x^{2}D - 504x \\
    L_{2} &= \textstyle\lclm\Bigl(D^{2},\bigl(x^{2} - x\bigr) D^{2} + \bigl(\frac{31}{24} x - \frac{5}{6}\bigr) D + \frac{1}{48}\Bigr).
  \end{align}
  The operator $L_{1}$ is the annihilator of the roots of $y^{5} + xy + 1$ in $K[y]$, so it only has algebraic solutions.
  The operator $L_{2}$ is the lclm of the operator from Example~\ref{ex:product_2F1} and $D^{2}$, so it has a transcendental solution and it has $D^{2}$ as a right factor.
  The order of the symmetric powers of the operators is growing as follows:
  \begin{center}
    \upshape
    \begin{tabular}[c]{rrrrrrr}
      $s$  & 1 & 2 & 3 & 4 & 5  \\
      \hline
      \rule[-4pt]{0pt}{14pt} $\ord(L_1^{\otimes s})$ & 4 & 9 & 15 & 21 & 27 \\
      \hline
      \rule[-4pt]{0pt}{14pt} $\ord(L_2^{\otimes s})$ & 4 & 10 & 20 & 35 & 56
    \end{tabular}
  \end{center}
  As predicted by the theorem, for $L_1$ the growth is linear, and for $L_2$ the growth is at least quadratic (cubic).
\end{ex}

The assumption on having $D^{2}$ as a right factor in the second part of the theorem cannot be dropped, as can be seen for example with $L=D^{2}-1$, whose solutions are $\exp(x)$ and $\exp(-x)$.
The solution space of $L^{\otimes s}$ is spanned by the terms $\exp(x(i - (s-i)))$ for $i \in \{0,\dots,s\}$, and therefore has dimension $s+1 = \O(s)$.
More generally, for any operator of order $r\leq 2$, the order of $L^{\otimes s}$ is bounded by $\binom{s+r-1}{s} \leq s+1$.
The divisibility condition says that $1$ and $x$ are solutions of~$L$, and in order to have in addition a transcendental solution, the order of $L$ must be at least~3.
If $L$ does not have $D^2$ as a right factor, apply the theorem to $\lclm(L,D^2)$ instead of~$L$.
Note that $L$ has only algebraic solutions if and only if $\lclm(L,D^2)$ has only algebraic solutions. 

More generally, if $M$ is any operator that has only algebraic solutions, then $L$ has only algebraic
solutions if and only if $\lclm(L,M)$ has only algebraic solutions. This is because, as remarked at
the end of Sect.~\ref{sec:prelim}, the least common multiple does not have any extraneous solutions. 
Nevertheless, as we show next, there is no hope that $\lclm(L,M)$ could have any pseudoconstants if
not already $L$ has any.

\begin{lem}
  \label{lem:lclm_pseudoconstants}
  Let $L,M \in K[D]$ and $N = \lclm(L,M)$.
  If $[P]_{N}$ is a nonzero completely integral element (resp. a pseudoconstant) in $K[D]/\langle N\rangle$, then
  at least one of $[P]_{L}$ or $[P]_{M}$ is a non-zero completely integral element (resp. a pseudoconstant) in the
  respective module. 
\end{lem}
\begin{proof}
  Let $[P]_{N}$ be a completely integral element of $K[D]/\langle N\rangle$.
  Let $E$ be an extension of $K$ such that $V(N) \subseteq E$ has dimension $\ord(N)$.

  Note that by definition of the lclm, both equivalence classes $[P]_{L}$ and $[P]_{M}$ are well-defined.
  Since $V(N) = V(L) + V(M)$, both $[P]_{L}$ and $[P]_{M}$ are completely integral.

  If $[P]_{N}$ is non-zero, there exists $h \in V(N)$ such that $P\cdot h \neq 0$.
  Therefore there exist $f \in V(L)$ and $g \in V(M)$ such that $h = f+g$ and $P\cdot f + P \cdot g \neq 0$.
  So at least one of $P \cdot f$ and $P \cdot g$ is nonzero, implying respectively that $[P]_{L}$ or $[P]_{M}$ is nonzero.

  The additional property that $P$ is not a constant similarly propagates to at least one of the summands.
\end{proof}

In view of this negative result, it is remarkable that taking symmetric products can produce
pseudoconstants. For example, the function considered in Example~\ref{ex:product_2F1} is a
product of an algebraic function and a hypergeometric function. The linear operator which
annihilates only the hypergeometric function (without the algebraic function multiplier)
does not have a pseudoconstant.
If the given operator $L$ has no pseudoconstants, we can thus ask whether there is an operator
$M$ with only algebraic solutions such that $L\otimes M$ has pseudoconstants.
Of course, as long as nobody tells us how to choose~$M$, this observation is not really helpful.
What we can easily do however is to multiply the solutions of $L$ with each other.
It turns out that this is sometimes sufficient.

\begin{ex}
  Consider the operator
  \[
    \textstyle
    L = \bigl(x^2 - x\bigr) D^{2} + \bigl(\frac{49}{6} x - \frac{7}{3}\bigr) D + 12
  \]
  annihilating the hypergeometric function ${}_{2}F_{1}\bigl(\frac{9}{2},\frac{8}{3}; \frac{7}{3}; x\bigr)$.
  The operator does not have a pseudoconstant.
  However, the operator $L^{\otimes 2}$ does have a pseudoconstant
  \begin{equation}
    \label{eq:9}
    \alpha(x) D^{2} + \beta(x) D + \gamma(x)
  \end{equation}
  where $\alpha$, $\beta$ and $\gamma$ are polynomials in $x$, with respective degree $11$, $10$ and $9$.
  By Theorem~\ref{thm:sympow_implies_trans} below, this implies that $L$ has at least one transcendental solution. 
\end{ex}

\begin{ex}
  \label{ex:product_2F1_noproduct}
  Consider the operator
  \[
    \textstyle
    L = \bigl(x^2 - x\bigr) D^{2} + \bigl(\frac{65}{24} x - \frac{7}{6}\bigr) D + \frac{35}{48}
  \]
  annihilating the hypergeometric function ${}_{2}F_{1}\bigl(\frac{7}{8},\frac{5}{6}; \frac{7}{6}; x\bigr)$.
  This is the hypergeometric function appearing in Example~\ref{ex:product_2F1}.
  
  The operator does not have a pseudoconstant.
  However, the operator $L^{\otimes 5}$ does have the pseudoconstant $[x(x-1)^3]$.
  By Theorem~\ref{thm:sympow_implies_trans} below, this implies that all nonzero solutions of $L$ are transcendental.

  The exponents of the solutions of $L$ at its singularities are:
  \begin{equation}
    \label{eq:5}
    \begin{array}{rcc}
      (0) & -\frac{1}{6} & 0\\[1ex]
      (1) & -\frac{13}{24} & 0\\[1ex]
      (\infty) & \frac{5}{6} & \frac{7}{8} 
    \end{array}
  \end{equation}
  Multiplying all the solutions by $x^{1/6}(x-1)^{13/24}$ allows us to clear the poles at $0$ and $1$, without creating a pole at infinity: the exponents at infinity become $\frac{5}{6}-\frac{1}{6}-\frac{13}{24} = \frac{1}{7}$ and $\frac{7}{8}-\frac{1}{6}-\frac{13}{24} = \frac{1}{8}$, both non-negative.
  This confirms the observation in Example~\ref{ex:product_2F1}.

  The presence of rational exponents in $x^{1/6}(x-1)^{13/24}$ means that it does not qualify as a pseudoconstant with our definition.
  However, considering symmetric powers allows us to clear those denominators.
  First, observe that the lowest exponents of the solutions of $L^{\otimes s}$ are $-\frac{1}{6}s$ at $0$, $-\frac{13}{24}s$ at $1$ and $\frac{5}{6}s$ at infinity.
  We are looking for a pseudoconstant of the form $[x^{a}(x-1)^{b}]$ with $a,b$ integers.
  Multiplying by such an element adds $a$ to the exponent at $0$, $b$ to the exponent at $1$, and subtracts $a+b$ from the exponent at infinity.
  The complete integrality condition thus translates into the following inequalities:
  \begin{align}
    \textstyle 0 \leq -\frac{1}{6}s + a  &&
    \textstyle 0 \leq -\frac{13}{24}s + b  &&
    \textstyle 0 \leq  \frac{5}{6}s - a -b. 
    \label{eq:8}
  \end{align}
  The solutions, for $s$ in $\{1,\dots,6\}$, are represented in Figure~\ref{fig:solutions_alg_mult}.
  The smallest value of $s$ for which there is an integer solution is $5$, and we recover the pseudoconstant $[x(x-1)^{3}] = [x^{4}-3x^{3}+3x^{2}-x]$ for $L^{\otimes 5}$. 
\end{ex}

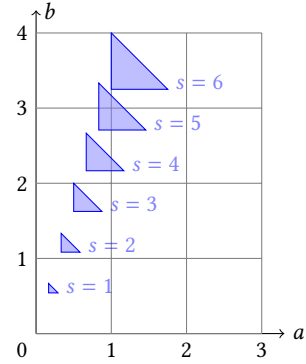
\begin{figure}
  \centering
    \upshape
    \begin{tikzpicture}[x=1cm,y=1cm]
      \draw[->] (0,0) -- (3.3,0) node[right] {$a$};
      \draw[->] (0,0) -- (0,4.3) node[right] {$b$};
      \draw[gray] (0,0) grid[step=1] (3,4);
      \foreach \a in {1, 2, 3} {\node[below] at (\a,0) {\a}; }
      \foreach \b in {1,...,4} {\node[left] at (0,\b) {\b}; }
      \node[below left] at (0,0) {0};
      
      \foreach \s in {1,..., 6} {
        \path[draw=blue, fill=blue!50!white, fill opacity=0.5, text=blue, text opacity=1]
        (\s*1/6,\s*13/24) 
        -- (\s*7/24,\s*13/24) node[anchor=base west] {$s=\s$}
        -- (\s*1/6,\s*4/6)
        -- cycle;
      }
    \end{tikzpicture}
    \caption{Solutions of the system~\eqref{eq:8} for $s$ in $\{1,\dots,6\}$}
    \label{fig:solutions_alg_mult}
\end{figure}

\begin{thm}
  \label{thm:sympow_implies_trans}
  Let $L \in K[D]$ be a differential operator.
  Suppose that for some $s\in\set N$ the symmetric power $L^{\otimes s}$ has a pseudoconstant.
  Then $L$ has at least one transcendental solution.
\end{thm}
\begin{proof}
  The solution space of $L^{\otimes s}$ is spanned by all products of $s$ solutions of $L$.
  The existence of a pseudoconstant in $K[D]/\langle L^{\otimes s} \rangle$ proves that at least one solution of $L^{\otimes s}$ is transcendental, and therefore at least one solution of $L$ is transcendental.
\end{proof}

In other words, a pseudoconstant for $L^{\otimes s}$ can be viewed as a transcendence certificate for~$L$.
As shown by the previous examples, such a certificate may exist even if $L$ itself does not have pseudoconstants.
So it is worthwhile to search for pseudoconstants of symmetric powers.
As shown by the following theorem, we cannot increase our chances to find a pseudoconstant any further by adding
some rational solutions to the solution space of~$L$.

\begin{prop}
  \label{prop:lclm_D_pseudoconstants}
  Let $M\in K[D]$ be an operator that has only solutions in~$K$, let $L\in K[D]$, and let $s\in\set N$.
  If $\lclm(L,M)^{\otimes s}$ has a pseudoconstant then there is a $d\in\{1,\dots,s\}$ such that $L^{\otimes d}$
  has a pseudoconstant.
\end{prop}
\begin{proof}
  First note that 
  \[
    L_s := \lclm(L,M)^{\otimes s} =
    \lclm\bigl(L^{\otimes s}, L^{\otimes (s-1)} \otimes M, \dots, M^{\otimes s}\bigr).
  \]
  By Lemma~\ref{lem:lclm_pseudoconstants}, if $[P]_{L_s}$ is a pseudoconstant, then there
  exists $d\in\{1,\dots,s\}$ such that $[P]_{L^{\otimes d}\otimes M^{\otimes(d-s)}}$ is
  also a pseudoconstant.

  This means that for every Puiseux series solution $f$ of $L$ at some point $\xi\in C\cup\{\infty\}$
  and every solution $r\in C(x)$ of $M$ we have that $P\cdot(r^{d-s}f^d)$ is integral, and
  that for at least one $r$ and one~$f$, the quantity $P\cdot(r^{d-s}f^d)$ is not a constant.
  Fixing one such solution $r\in C(x)\setminus\{0\}$ of~$M$, it follows that $Pr^{d-s}$ is
  a completely integral element of $K[D]/\<L^{\otimes d}>$ and that $[Pr^{d-s}]_{L^{\otimes d}}$
  is not a constant. Thus $L^{\otimes d}$ has the pseudoconstant $[Pr^{d-s}]_{L^{\otimes d}}$.
\end{proof}

We have not been able to answer the following question: 

\begin{question}\label{q}
  Is it true that for every operator $L$ with at least one transcendental solution there exists
  an $s\in\set N$ such that $L^{\otimes s}$ has a pseudoconstant?
\end{question}

If the answer to Question~\ref{q} is yes, then this fact in combination
with Alg.~\ref{alg:algsols} would yield a new decision procedure for the existence of transcendental solutions.
We could simply search in parallel for $s=1,2,3,\dots$ for an algebraic solution of $L$ of degree $s$
and a pseudoconstant of $L^{\otimes s}$. Exactly one of these parallel threads would have to terminate
after a finite number of steps.

A natural idea to prove the existence of pseudoconstants of $L^{\otimes s}$
for sufficiently large~$s$ is to show that the linear system, which emerges
from a search for pseudoconstants via the linear algebra approach,
has more variables than equations for sufficiently large~$s$.
Unfortunately, this does not seem to be the case: indeed, if $R(s)$ is the order of $L^{\otimes s}$, the ansatz~\eqref{eq:1} has $\Theta(NR(s))$ undetermined coefficients.
As for the number of equations, it is equal to the number of series coefficients to set to zero: for each series solution $f_{i}$ ($i \in \{1,\dots,R(s)\}$), the valuation of $q(f_{i})$ can be as low as $-N$, for a total of $\Theta(NR(s))$ equations.

The following example can perhaps be considered as some piece of empirical evidence that the
answer to Question~\ref{q} is no.
On the other hand, we can show (Prop.~\ref{prop:alg-constants}) that for an operator $L$ with only algebraic
solutions there is always an $s$ such that $L^{\otimes s}$ has a constant (but of course no
pseudoconstant), and this could be considered as some piece of evidence that the answer to
Question~\ref{q} may be yes.

\begin{ex}
  Consider the operator
  \[
  \textstyle
  \bigl(x^{2} - x\bigr) D^{2} + \bigl(\frac{164}{15} x - \frac{16}{3}\bigr) D + \frac{1403}{60},
  \]
  which annihilates the hypergeometric function ${}_{2}F_{1}\bigl(\frac{61}{10},\frac{23}{6}; \frac{16}{3}; x\bigr)$.
  Thanks to Schwarz' classification, we know that the operator has no algebraic solutions. 
  However, an exhaustive search using integral bases could not find a completely integral element
  for $L^{\otimes s}$ for any $s\leq 6$, and a heuristic search using linear algebra could not
  find one for any $s\leq 30$.
\end{ex}

\begin{lem}\label{lem:ratsolmakesconstant}
  Let $M\in K[D]$ and let $q\in K$ be such that $M\cdot q\neq0$.
  Then $L:=\lclm(qD-q',M)$ has a nonzero constant.
\end{lem}
\begin{proof}
  Note that $V(L)=\Span(q)+V(M)$ and $u:=M\cdot q\neq0$. Consider $P:=u^{-1}M$.
  Every $f\in V(L)$ can be written as $f=cq+m$ for a $c\in C$ and an $m\in V(M)$.
  So $P\cdot f=u^{-1}(M\cdot m+cM\cdot q)=u^{-1}cu=c$.
  By Prop.~\ref{prop:constant} part~\ref{prop:constant:1}, it follows that $[P]$ is
  a nonzero constant of~$L$.
\end{proof}

\begin{prop}\label{prop:alg-constants}
  If $L\in K[D]$ has only algebraic solutions and
  $d$ is such that all the solutions of $L$ have a minimal polynomial of degree
  at most~$d$, then $L^{\otimes d}$ has a nonzero constant.
\end{prop}
\begin{proof}
  Since $L$ has only algebraic solutions, also $L^{\otimes d}$ has only algebraic solutions.
  Moreover, $L^{\otimes d}$ has at least one nonzero rational function solution~$q$
  (e.g., the product of all the conjugates of some algebraic solution of~$L$).
  If $f$ is a solution of $L^{\otimes d}$, then so are all the conjugates of~$f$,
  because $L^{\otimes d}$ has coefficients in~$K$.
  The solution space of the minimal order annihilating operator of $f$ is generated
  by $f$ and its conjugates and therefore a right factor of~$L^{\otimes d}$.

  Let $f_1$ be a solution of $L^{\otimes d}$ which does not belong to $\Span(q)$,
  and let $M_1$ be a minimal order annihilating operator of~$f_1$.
  For $n=1,2,\dots$, let $f_n$ be a solution of $L^{\otimes d}$ which does not belong to $\Span(q)+V(M_1)+\cdots+V(M_{n-1})$,
  and let $M_i$ be a minimal order annihilating operator of~$f_n$, until we have
  $V(L^{\otimes d})=\Span(q)+V(M_1)+\cdots+V(M_n)$.
  At this stage, we have
  \[
  L^{\otimes d}=\lclm(qD-q',\lclm(M_1,\dots,M_n)),
  \]
  and since $\lclm(M_1,\dots,M_n)\cdot q\neq0$ by the choice of $M_1,\dots,M_n$,
  Lemma~\ref{lem:ratsolmakesconstant} applies.
  The claim follows.
\end{proof}

\section{Conclusion}

We propose the notion of a \emph{transcendence certificate} for any kind of artifact
whose existence implies that a given differential operator has at least one
transcendental solution.  Simple transcendence certificates are logarithmic and
exponential singularities. \emph{Pseudoconstants} introduced in
Def.~\ref{def:pseudoconstants} can also serve as transcendence certificates. We
have given examples of operators that have no logarithmic or exponential
singularities but that do have pseudoconstants.

We have also given examples of operators that have no pseudoconstants even
though they have transcendental solutions. To such operators, we can try to
apply transformations that preserve the existence of transcendental solutions
but may lead to the appearance of pseudoconstants. In particular, as shown in
Sect.~\ref{sec:expand-search-space}, it can happen that an operator $L$ has no
pseudoconstants but some symmetric power $L^{\otimes s}$ of $L$ does. A
pseudoconstant of $L^{\otimes s}$ suffices to certify the existence of a
transcendental solution of~$L$. An open question (Question~\ref{q}) is whether
the existence of transcendental solutions of $L$ implies the existence of an $s$
such that $L^{\otimes s}$ has pseudoconstants. We would be very interested in an
answer to this question.

There are further possibilities to transform an operator with no pseudoconstants
to one that may have some. For example, we could try to exploit that the
composition of a D-finite function with an algebraic function is always D-finite.
If $f$ is D-finite and $g$ is algebraic, then $f\circ g$ is algebraic if and
only if $f$ is algebraic, thus a pseudoconstant for an annihilating operator
of $f\circ g$ could serve as a transcendence certificate for an annihilating
operator of~$f$. Note that unlike the transformations considered in this paper,
the composition can not only remove singularities but also create new ones.
We have not found an example where this process reveals new pseudoconstants.

In another direction, we could try to weaken the requirements of Def.~\ref{def:pseudoconstants}. According
to our definition, $[P]_L$ is a pseudoconstant if \emph{every} local solution $f$ of $L$
is such that $P\cdot f$ has nonnegative valuation. For a transcendence certificate,
it would suffice to have \emph{one} global solution $f$ of $L$ (a complex
function defined on a Riemann surface) which is not constant and has no pole.
If we relax Def.~\ref{def:pseudoconstants} accordingly, it may be that additional
operators would have pseudoconstants. However, we would no longer know how to decide
the existence of pseudoconstants for a given operator.

\paragraph{Acknowledgments}
We are grateful to Alin Bostan and Bruno Salvy for their valuable feedback on
the topic of this paper, after a talk at JNCF 2023. We also thank the
anonymous referees for their suggestions to improve the paper.


\end{document}